\documentclass[pra,twocolumn,floatfix,aps,superscriptaddress,showpacs,amsfonts,10pt]{revtex4}
\usepackage{graphicx}
\usepackage{array, color}
\usepackage{amsthm}
\usepackage{amsmath}
\usepackage{amssymb}
\newtheorem{theorem}{Theorem}

\begin{document}
\title{Coherence distribution in multipartite systems}

\author{Zhengjun Xi}
\email{xizhengjun@snnu.edu.cn}
\affiliation{College of Computer Science, Shaanxi Normal University, Xi'an 710062,
China}

\date{\today}

\begin{abstract}
This paper examines the coherence in multipartite systems. We first discuss the distribution of total coherence in a given multipartite quantum state into discord between subsystems and
coherent dissonance in each individual subsystem, using the relative entropy as a distance measure. Then
we give some trade-off relations between various types of coherence and discord within a bipartite system,
 and extend these results to the multipartite setting. Finally, the change of coherence in entanglement distribution is studied.
\end{abstract}
\maketitle
\section{Introduction}
Quantum coherence arising from quantum superposition plays a central role in
quantum mechanics, and it is also a common necessary condition for
entanglement and other types of quantum correlations such as Bell nonlocality and quantum discord.
Recently, researchers have begun to develop a resource-theoretic framework for
understanding quantum coherence~\cite{Aberg06,Baumgratz13}; for more discussions we refer to~\cite{Hu17,Streltsov17}.

It is well-known that nonclassical correlation (e.g., entanglement, discord) between subsystems is a form of coherence, but the converse is not necessarily true:
 coherence may also exist in individual subsystems. As suggested by Modi's $et$ $al$~\cite{Modi10}, for a quantum state $\rho$, its nonclassical correlation is measured by the relative entropy of discord, which indicates the distance of $\rho$ to the nearest classical state $\chi_\rho$. Since coherence is a basis-dependent quantity, the nearest classical state $\chi_\rho$ may not be diagonalizable under the reference incoherent states. This means that $\chi_\rho$ may possess some coherence located in the subsystems. More extremely, product states in which no nonclassical correlation exists at all can be maximally coherent. For example, the product state $|+\rangle |+\rangle$ where $|+\rangle = \frac 1{\sqrt{2}} (|0\rangle+|1\rangle)$ on a two qubit system is a maximally coherent state  with the basis $\{|ij\rangle\}^1_{i,j=0}$. Despite the considerable attempts to understand this phenomenon~\cite{Streltsov15, Adesso15,Yao15, Hu17a, Xi15, Ma16,Radhakrishnan16,Chitambar16}, one of the most important questions remains unresolved: How can we quantify the collective coherence and local coherence, and characterize the relations between them and the nonclassical correlation in a multipartite system? The answer to this question is also of importance for the analysis of quantum algorithms where coherent superposition is often essential.

In this paper, we examine the distribution of total coherence in a given multipartite quantum state into discord and
coherent dissonance using the relative entropy as a distance measure. Trade-off relations between various types of coherence and discord within both bipartite and multipartite systems are also considered.
The paper is organized as follows. In Sec.~\ref{sec:ardc}, we recall a unified measure to quantify different correlations like discord, coherence, and coherent dissonance of multipartite quantum states.
In Sec.~\ref{sec:com}, we give some relations between these correlations for bipartite systems, when local measurements on one party are considered.
These results are then generalized to the multipartite setting. We discuss the change of coherence in entanglement distribution in Sec.~\ref{sec:ms}, and
Sec.~\ref{sec:con} is devoted to a brief conclusion.

\section{Additivity relations for discord and coherence}\label{sec:ardc}
Following the framework in~\cite{Modi10}, we employ relative entropy as a measure of distance to characterize the separation of total coherence into discord and coherent dissonance. For an $N$-partite system $\mathcal{H}$, let $\{|k_i^j\rangle : 1\leq j\leq n_i\}$ be a fixed basis for system $i$ where $1\leq i\leq N$, and $\mathcal{K}:= \{|\vec{k}\rangle\}$ the product basis of $\mathcal{H}$ induced by them.
The set of incoherent states with respect to $\mathcal{K}$, denoted $\mathcal{I}^{\mathcal{K}}_N$, contains all locally distinguishable states
\begin{equation}
\sum_{|\vec{b}\rangle\in \mathcal{K}}p_{\vec{k}}|\vec{k}\rangle\langle\vec{k}|,
 \end{equation}
 where $\{p_{\vec{k}}\}$ is a joint probability distribution.  Note that $\mathcal{I}^{\mathcal{K}}_N$ is a convex and compact set.
Furthermore, a state is said to be classical, if it is incoherent with respect to some product basis of $\mathcal{H}$.
We denote the set of classical states by $\mathcal{I}_N$. Obviously,
\begin{equation}
\mathcal{I}_N :=\bigcup_{\mathcal{K}}\mathcal{I}^{\mathcal{K}}_N.
\end{equation}

Let  $\rho$ be an $N$-partite state, and $\mathcal{K}$ a product basis of $\mathcal{H}$. The \emph{coherence} of $\rho$
with respect to $\mathcal{K}$ can be measured by the relative entropy of discord,
which is the distance between $\rho$ and the nearest classical state in $\mathcal{I}^{\mathcal{K}}_N$. That is,
\begin{equation}\label{def:coh}
C^\mathcal{K}(\rho):=\min_{\delta\in\mathcal{I}^{\mathcal{K}}_N}S(\rho||\delta),
\end{equation}
We denote by $\delta_\rho^{\mathcal{K}}:=\sum_{|\vec{k}\rangle\in \mathcal{K}}\langle\vec{k}|\rho|\vec{k}\rangle|\vec{k}\rangle\langle\vec{k}|$ this nearest state achieving the minimum in Eq.(\ref{def:coh}). It is well known that coherence has a nice closed expression, given by the entropy change caused by the dephasing operation on the state:
\begin{equation}\label{eq:coh_e}
C^\mathcal{K}(\rho)=S(\delta^\mathcal{K}_\rho)-S(\rho).
\end{equation}

On the other hand, we define the \emph{discord} of $\rho$ as the nonclassical correlation between different subsystems, which is characterized by the
smallest distance between $\rho$ and any classical state. That is,
\begin{equation}\label{def:K_discord}
Q(\rho):=\min_{\chi\in\mathcal{I}_N}S(\rho||\chi).
\end{equation}
Note that the set $\mathcal{I}_N$ is not a convex set; mixing two classical states written in different bases can give rise to a nonclassical state.
Since the minimization of the relative entropy over
classical states is identical to the minimization of the entropy $S(\chi)$ over the
choice of local basis $|\vec{k}\rangle$, Modi $et.$ $al.$ gave a useful expression~\cite{Modi10}
\begin{equation}\label{eq:discord_e}
Q(\rho)=S(\chi_\rho)-S(\rho),
\end{equation}
where $S(\chi_\rho)=\min_{\{|\vec{k}\rangle\}}S(\sum_{\vec{k}}|\vec{k}\rangle\langle\vec{k}|\rho|\vec{k}\rangle\langle\vec{k}|)$.
Note that discord is independent of the choice of basis, from the point  of view of~\cite{Radhakrishnan16},
it is intrinsic coherence.

We denote by $\chi_\rho$ this nearest state achieving the minimum in Eq.(\ref{def:K_discord}). Note that  $\chi_\rho$ may not be diagonalizable under the basis $\mathcal{K}$, thus it may still have some coherence (with respect to $\mathcal{K}$) inside the individual subsystems, which can be captured with the help of coherence just defined above:
\begin{equation}
D^\mathcal{K}(\rho):= C^\mathcal{K}(\chi_\rho) = \min_{\delta\in\mathcal{I}^{\mathcal{K}}_N}S(\chi_\rho||\delta).
\end{equation}
We call it coherent dissonance, which is similar to the quantum dissonance defined in~\cite{Modi10}, where quantum dissonance is defined as nonclassical correlations which exclude entanglement.
From Eq.(\ref{eq:coh_e}), we have
\begin{equation}\label{eq:diss_e}
D^\mathcal{K}(\rho)=S(\delta_{\chi_\rho}^\mathcal{K})-S(\chi_\rho).
\end{equation}
Note that the dissonance defined here is different from the local coherence in~\cite{Radhakrishnan16} which is the relative entropy between the
nearest separable state and the nearest incoherent state of $\rho$.

The results above give us a method to compute discord and coherence, they also give us the following additivity
relations:
\begin{theorem}\label{theorem_qcd_1}
For any $N$-partite quantum state $\rho$, the following inequality holds:
\begin{equation}\label{eq:dist_cc}
Q(\rho)\leq C^{\mathcal{K}}(\rho)\leq Q(\rho)+D^{\mathcal{K}}(\rho).
\end{equation}
\end{theorem}
\begin{proof}
The first inequality is obvious. Now
from Eqs.~(\ref{eq:coh_e}), ~(\ref{eq:discord_e}) and~(\ref{eq:diss_e}), we have
\begin{align}
C^{\mathcal{K}}(\rho)=&\ S(\delta^{\mathcal{K}}_\rho)-S(\rho)\nonumber\\
\leq&\ S(\chi_\rho)-S(\rho)+S(\delta^{\mathcal{K}}_{\chi_\rho})-S(\chi_\rho)\nonumber\\
=&\ Q(\rho)+C^{\mathcal{K}}(\rho),
\end{align}
where the inequality comes from the definition of coherence.
\end{proof}

We denote the quantity $L(\rho):=S(\delta^{\mathcal{K}}_{\chi_\rho})-S(\delta^{\mathcal{K}}_\rho)$;
 it describes the entropic costs caused by the optimal dephasing measurement
with respect to the discord. Clearly, we have $L(\rho)\geq 0$, and also
\begin{equation}
C^{\mathcal{K}}(\rho)+L(\rho)=Q(\rho)+D^{\mathcal{K}}(\rho).
\end{equation}
This relation corresponds to the closed path in Fig.~\ref{Fig1}
and means that the sum of the nonclassical correlation
and coherent dissonance is equal to the sum of the total coherence and the entropic costs.

To conclude this section, we provide two examples to illustrate the different contributions for the coherence. The first one is a maximally coherent state $|+\rangle|+\rangle$ on a two qubit system, where $|+\rangle=\frac{1}{\sqrt{2}}(|0\rangle+|1\rangle)$ is a maximally coherent state on one qubit system. Obviously, we have $C^{\mathcal{K}}=D^{\mathcal{K}}=2$, but $Q=L=0$. This implies that the coherence comes solely from the coherence located in the individual qubits. The second example is a Bell state $|\psi\rangle=\frac{1}{\sqrt{2}}(|00\rangle+|11\rangle)$. It is easy to check that $C^{\mathcal{K}}=Q=1$, and $D^{\mathcal{K}}=L=0$. This means that the coherence comes solely from the correlation between subsystems.
\begin{figure}
\centering\includegraphics[width=1.5in]{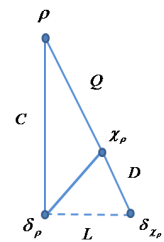}
\caption{Quantum coherence in a multipartite state $\rho$. The total coherence $C^{\mathcal{K}}(\rho)$
comes from quantum correlation $Q(\rho)$ and coherent dissonance $D^{\mathcal{K}}(\rho)$.
With the aid of $L(\rho)$ the closed path is additive, and has a clear operational interpretation. In general,
we have $D^{\mathcal{K}}(\rho)\leq S(\chi_\rho||\delta^{\mathcal{K}}_\rho)$.}\label{Fig1}
\end{figure}

\section{Additivity relations via local measurement}\label{sec:com}
Most of nonclassical correlations measures are limited
to studies of bipartite correlations only as the original concept of
discord, which involves bipartite
system with classicality for only one subsystem.In this section, we will discuss the connection between nonclassical correlations and coherence in the bipartite system by the act of the local measurement on one or both subsystems. Let $\rho^{AB}$
be a quantum state in a bipartite system $AB$, and $\{|a\rangle^A\}$ be a given basis of $A$.
The quantum incoherent relative entropy of $\rho$ with respect to $\{|a\rangle^A\}$ is defined as
\begin{equation}
C^{A|B}(\rho^{AB}):=\min_{\sigma^{AB}\in \mathcal{I}^a}S(\rho^{AB}||\sigma^{AB}).
\end{equation}
Here $\mathcal{I}^a$ is the set of all quantum incoherent states
of the form, i.e., $\sum_{a}p_a|a\rangle^A\langle a|\otimes\rho^B_a$,
where $\rho^B_a$ are quantum states on $B$, and $\{p_a\}$ is a probability distribution~\cite{Chitambar16}.
It has been proved that the quantum incoherent relative entropy can be written as
\begin{equation}\label{eq:qire}
C^{A|B}(\rho^{AB})=S(\sigma^{AB}_\rho)-S(\rho^{AB}),
\end{equation}
where $\sigma^{AB}_\rho=\sum_a|a\rangle^A\langle a|\rho^{AB}|a\rangle^A\langle a|$.

Furthermore, when the basis $\{|a\rangle^A\}$ varies, we obtain the set of all classical-quantum states, denoted $\mathcal{I}$,
clearly, $\mathcal{I}_2\subset\mathcal{I}$.
Then, the one-way quantum discord~\cite{Horodecki05} is defined as
\begin{equation}\label{def:1wqd}
Q^{A|B}(\rho^{AB})=\min_{\omega^{AB}\in\mathcal{I}}S(\rho^{AB}||\omega^{AB}).
\end{equation}
Similar to Eq.~(\ref{eq:discord_e}), the one-way quantum discord has also a useful expression, i.e.,
\begin{equation}\label{eq:1discord_e}
Q^{A|B}(\rho^{AB})=S(\omega^{AB}_\rho)-S(\rho^{AB}),
\end{equation}
where $S(\omega^{AB}_\rho)=\min_{\{|i\rangle^A\}}S(\sum_{i}|i\rangle^A\langle i|\rho^{AB}|i\rangle^A\langle i|)$,
and the minimization is taken over all orthogonal bases of subsystem $A$.
Note that for the nearest classical-quantum state $\omega^{AB}_\rho$, its reduced state on $A$
may not be diagonalizable under the basis $\{|a\rangle\}$, thus it may still have some coherence, which can be captured with the help of coherence just defined above:
\begin{equation}
D^{A|B}(\rho^{AB}):=C^{A|B}(\omega^{AB}_\rho)=\min_{\sigma^{AB}\in\mathcal{I}^a}S(\omega^{AB}_\rho||\sigma^{AB}).
\end{equation}
Similar to coherent dissonance, we call it the one-way
coherence dissonance.
Then Eq.~(\ref{eq:qire}) yields a closed expression for one-way coherence dissonance
\begin{equation}\label{eq:1diss_e}
D^{A|B}(\rho^{AB})=S(\sigma^{AB}_{\omega_\rho})-S(\omega^{AB}_\rho).
\end{equation}
Equipped with these relations, we are now in a position to show a close connection between the one-way quantum discord $Q^{A|B}$,
quantum incoherent relative entropy $C^{A|B}$ and the one-way coherence dissonance $D^{A|B}$, i.e.,
\begin{equation}\label{eq:1dist_cc}
Q^{A|B}(\rho^{AB})\leq C^{A|B}(\rho^{AB})\leq Q^{A|B}(\rho^{AB})+D^{A|B}(\rho^{AB}).\nonumber
\end{equation}
This relation is similar to Theorem~\ref{theorem_qcd_1}, and provides an upper bound of quantum incoherent relative entropy.

Next, we compare the coherence and the discord defined in~\cite{Zurek01}. From Eq.~(\ref{eq:qire}),
after some simple algebraic computation, we obtain a fundamental inequality as follows:
\begin{equation}\label{eq:c_qire_c}
C^{A|B}(\rho^{AB})\leq C(\rho^{AB})-C(\rho^B).
\end{equation}
This relation has a clear operational interpretation: the difference of coherence between the total
system and one of its subsystems is no less than the quantum incoherent relative entropy
via the other subsystem.

From~(\ref{eq:qire}), it is easy to see the following inequality:
\begin{equation}\label{eq:c_qire_c_2}
\Theta^{A|B}(\rho^{AB})+C(\rho^A)\leq C^{A|B}(\rho^{AB}),
\end{equation}
where $\Theta^{A|B}(\rho^{AB})=S_{\rho^{AB}}(A:B)-\min_{\{|i\rangle^A\}}S_{\rho^{AB}_{|i\rangle^A}}(A:B)$
is another definition of discord in~\cite{Zurek01}.
Here $S_{\rho^{AB}}(A:B)$ is the quantum mutual information of $\rho^{AB}$, and $\rho^{AB}_{|i\rangle^A}=\sum_i|i\rangle^A\langle i|\rho^{AB}|i\rangle^A\langle i|$ for any
orthonormal basis $\{|i\rangle^A\}$ of $A$.
Then, we have the following inequality, which is tighter than the subadditivity of coherence in~\cite{Xi15}.
\begin{theorem}\label{thm:thm2}
For any bipartite quantum state $\rho^{AB}$, the following inequality holds:
\begin{equation}\label{eq:c_qire_c_3}
\Theta^{A|B}(\rho^{AB})+C(\rho^A)+C(\rho^B)\leq C(\rho^{AB}).
\end{equation}
\end{theorem}
\begin{proof}
This is direct from Eqs.~(\ref{eq:c_qire_c}) and (\ref{eq:c_qire_c_2}).
\end{proof}
The above theorem shows that the total coherence of a quantum state includes genuine local coherence located in the individual subsystems
and quantum
correlation between them.

To illustrate the inequality presented in Theorem~\ref{thm:thm2}, let us consider two simple examples. The first one is a separable state with the reduced states both being maximally mixed~\cite{Datta08}:
\begin{align}
\rho^{AB}= \frac{1}{4}[&\ |+\rangle\langle+|\otimes|0\rangle\langle 0|+|-\rangle\langle-|\otimes|1\rangle\langle 1|\nonumber\\
&\ |0\rangle\langle 0|\otimes |-\rangle\langle -|+|1\rangle\langle 1|\otimes |+\rangle\langle +|\ ].\nonumber
\end{align}
It is easy to show that $\Theta^{A|B}(\rho^{AB})\thickapprox0.311$, $C(\rho^{AB})=0.5$, and $C(\rho^A)=C(\rho^B)=0$.
The second example is the Werner state
\begin{equation}
\rho^{AB}=(1-p)\frac{I}{4}+p|\psi\rangle\langle\psi|,\nonumber
\end{equation}
where $|\psi\rangle=\frac{1}{\sqrt{2}}(|00\rangle+|11\rangle)$ is a Bell state, and $p\in[0,1]$. We know that its
discord is greater than $0$ when $p>0$
and it is separable when $p\leq \frac{1}{3}$.
Clearly, we have $C(\rho^A)=C(\rho^B)=0$, and from~\cite{Luo08},
the nearest classical state of $\rho^{AB}$ is the closet incoherent state, and then
$\Theta^{A|B}(\rho^{AB})=C(\rho^{AB})$. That is to say, for Bell-diagonal states, the equality holds in Eq.(\ref{eq:c_qire_c_3}). However, the question remains open for general mixed states.

The discord above is defined to be independent of
measurements by requiring the optimization over all
measurements. However, some references
where only a particular measurement is relevant, and one can define the measurement dependent
discord. In this situation, we can provide a more tighter form for the relation~(\ref{eq:c_qire_c_3}). But, we know the discord does not involve any optimization, this results in overestimation of the amount of nonclassical correlations~\cite{Modi12}, and it involves no optimization it
is not a particularly good measure of correlations~\cite{Brodutch12}.

To provide a tighter low bound for the relation~(\ref{eq:c_qire_c_3}), we try to consider the symmetric discord~\cite{Wu09}. Recall that the symmetric discord is defined as
\begin{equation}\label{eq:two_discord_def}
\Theta(\rho^{AB})=S_{\rho^{AB}}(A:B)-\min_{\{|ij\rangle^{AB}\}}S_{\rho^{AB}_{\{|ij\rangle^{AB}\}}}(A:B),
\end{equation}
where $\rho^{AB}_{\{|ij\rangle^{AB}\}}=\sum_{ij}|ij\rangle^{AB}\langle ij|\rho^{AB}|ij\rangle^{AB}\langle ij|$,
and the set $\{|ij\rangle^{AB}\}$ constitutes a product orthonormal basis of $AB$. Clearly, we have $\Theta^{A|B}(\rho^{AB})\leq\Theta(\rho^{AB})$.
The following theorem presents an improved lower bound for the coherence.
\begin{theorem}
For any bipartite quantum state $\rho^{AB}$, the following inequality holds:
\begin{equation}\label{eq:c_qire_c_4}
\Theta(\rho^{AB})+C(\rho^A)+C(\rho^B)\leq C(\rho^{AB}).
\end{equation}
\end{theorem}
\begin{proof}
 Without loss of generality, we assume that $\{|ij\rangle^{AB}\}$ is the fixed product basis on $AB$. Then, we perform a measurement on $AB$ with respect to this basis, from the result~\cite{Xi15}, we have
\begin{equation}
C(\rho^{AB})-C(\rho^A)-C(\rho^B)=S_{\rho^{AB}}(A:B)-S_{\rho^{AB}_{|ij\rangle^{AB}}}(A:B).\nonumber
\end{equation}
Taking the optimization over all the orthonormal measurements on $AB$ completes the proof of the theorem.
\end{proof}

Note that quantum mutual information for bipartite quantum systems is non-negative, and
is viewed as the total correlation between the two subsystems. However, it is no longer
true for three or more party quantum systems, there are at least two different definitions for quantum mutual information, it is possible that three-party quantum mutual information can be negative~\cite{Kumar17}. In this paper, we take Modi's suggestion~\cite{Modi10}, for an $N$-partite quantum state, its quantum mutual information is defined by the relative entropy between
it and the product state obtained from its reduced states, and bears the interpretation of total correlation between all the subsystems. Thus, we now extend our results Eq.~(\ref{eq:c_qire_c_2}) and~(\ref{eq:c_qire_c_4}) to the multipartite setting.

\begin{theorem}\label{theorem_qcd}
For any $N$-partite quantum state $\rho^{12\cdots N}$, we have
\begin{equation}\label{eq:multi_dis}
\sum_{i=1}^{N-1}\Theta^{i|i+1\cdots N}(\rho^{i i+1\cdots N})+\sum_{i=1}^N C(\rho^i)\leq C(\rho^{12\cdots N})
\end{equation}
and
\begin{equation}\label{eq:multi_dis_5}
\Theta(\rho^{12\cdots N})+\sum_{i=1}^N C(\rho^i)\leq C(\rho^{12\cdots N}),
\end{equation}
where $\Theta^{i|i+1\cdots N}(\rho^{i i+1\cdots N})$ is the quantum discord based the
measurement on the subsystem $i$, $\Theta(\rho^{12\cdots N})$ is a generalization of quantum discord in Eq.~(\ref{eq:two_discord_def}), and $C(\rho^i)$ is relative entropy of coherence of the reduced system $i$.
\end{theorem}

\section{Coherence in entanglement distribution}\label{sec:ms}
In the previous section, we have given a clear operational interpretation for Eq.~(\ref{eq:c_qire_c}).
Here, we discuss the general case via entanglement distribution.
The general scenario for entanglement distribution is consider in~\cite{Cubitt03,Streltsov12,Chuan12}.
We assume that two agents, Alice and Bob, have access to a tripartite quantum state $\rho^{ABR}$, with Alice holding $A$ and $R$, and Bob holding $B$.
The entanglement distribution is realized by sending the particle $R$ from Alice to Bob.
If the quantum channel used for the transmission is noiseless,
the amount of entanglement distributed in this process is quantified by the difference
between the final amount of entanglement $E^{A|BR}$ and the initial amount of entanglement $E^{AR|B}$.
It has be proven in~\cite{Streltsov12,Chuan12} that
 the amount of distributed entanglement cannot exceed the quantum
 discord between $R$ and the remaining systems $AB$.
Then, we will give a bound on the increase of coherence during the entanglement distribution.
\begin{theorem}
Given a tripartite state $\rho:=\rho^{ABR}$, the following inequality holds:
\begin{equation}\label{eq:dist_c}
C^{AR|B}(\rho)-C^{A|BR}(\rho)\leq C^{R|AB}(\rho).
\end{equation}
\end{theorem}
\begin{proof}
Let $\{|j\rangle^A\}$ and $\{|i\rangle^R\}$ be the fixed basis on $A$ and $R$, respectively.
Then, we denote the states $$\rho^\prime=\sum_i|i\rangle^R\langle i|\rho|i\rangle^R\langle i|$$
and
$$\rho^*=\sum_{i,j}(|j\rangle^A\langle j|\otimes|i\rangle^R\langle i|)\rho(|j\rangle^A\langle j|\otimes|i\rangle^R\langle i|)$$
to arise from $\rho$ via the local orthonormal  measurements $\{|i\rangle^R\langle i|\}$
on the subsystem $R$ and the local orthonormal measurement $\{|j\rangle^A\langle j|\otimes|i\rangle^R\langle i|\}$
on the subsystem $AR$, respectively. This implies that the states $\rho^\prime$ and $\rho^*$
are the nearest quantum incoherent states for the sake of $C^{R|AB}(\rho)$ and $C^{AR|B}(\rho)$, respectively.
Then we have
\begin{align}
C^{AR|B}(\rho)=&\ [S(\rho^\prime)-S(\rho)]+[S(\rho^*)-S(\rho^\prime)]\nonumber\\
=&\ C^{R|AB}(\rho)+C^{A|BR}(\rho^\prime)\nonumber\\
\leq&\ C^{R|AB}(\rho)+C^{A|BR}(\rho),
\end{align}
where the inequality comes from the fact that quantum incoherent relative entropy does
not increase under the action of incoherent operation.
\end{proof}
We will now consider the situation where the total state is pure.
One applies Eq.~(\ref{eq:dist_c}) to a
tripartite pure state $|\psi\rangle^{ABR}$, it reduces to
\begin{equation}
S(\tilde{\rho}^{AR})\leq S(\tilde{\rho}^A)+S(\tilde{\rho}^R),
\end{equation}
where the dephased state $\tilde{\rho}^{AR}=\sum_{i,j}(|j\rangle^A\langle j|\otimes|i\rangle^R\langle i|)\rho^{AR}(|j\rangle^A\langle j|\otimes|i\rangle^R\langle i|)$ with the reduced state $\rho^{AR}=\mathrm{Tr}_B(|\psi\rangle^{ABR}\langle \psi|)$.
This is the additivity of entropy for subsystems $A$ and $R$ after the measurements.

Finally, we consider the more general situation in which the channel used for entanglement distribution is noisy.
If Alice uses an incoherent channel $\Lambda^R$ to send her particle $R$ to Bob, they end up in the final state $\rho_f=\Lambda^R(\rho_i)$, where $\rho_i=\rho^{ABR}$ is a tripartite initial state.
In the following theorem we show the amount of coherence in the entanglement distribution.
\begin{theorem}
Given a quantum incoherent channel $\Lambda^R$  and two states $\rho_i$ and $\rho_f$, the following inequality holds:
\begin{equation}\label{eq:dist_c1}
C^{AR|B}(\rho_f)-C^{A|BR}(\rho_i)\leq C^{R|AB}(\rho_f).
\end{equation}
\end{theorem}
\begin{proof}
We first apply Eq.~(\ref{eq:dist_c})
to the state $\rho_f$, deriving  $C^{AR|B}(\rho_f)-C^{A|BR}(\rho_f)\leq C^{R|AB}(\rho_f)$. One then completes the proof by noting $C^{A|BR}(\rho_f)\leq C^{A|BR}(\rho_i)$, which follows from the fact that the quantum incoherent relative entropy does not increase under the action of incoherent operation.
\end{proof}
\section{conclusions}\label{sec:con}
We explored the distribution of coherence in multipartite systems. We studied the separation of total coherence in a given quantum state into quantum correlations and
coherent dissonance using the relative entropy as a distance measure, and an additivity relation between them is given. Then,
 some trade-off relations between various types coherence and the discord within the bipartite are given, and we extended our results to multipartite setting. We also discussed the amount of change of coherence in the entanglement distribution by having access to a noiseless and noisy quantum incoherent channel.
Our results have direct importance in the theoretical description of coherence and practical application in quantum algorithms where the key steps involve the coherent superposition. We also hope our results can be used to find the optimal quantum resources in quantum communication tasks.

\emph{Note added}$-$During the completion of this paper we
became aware of the closely related independent work by Bu $et$ $al$~\cite{Bu17}.

\section{Acknowledgments}
Z. Xi is supported by the National Natural Science Foundation of China (Grant Nos. 61671280, and 11531009), and by the Fundamental Research Funds for the Central Universities (GK201502004),
and by the Funded Projects for the Academic Leaders and Academic Backbones, Shannxi Normal University (16QNGG013).


\begin{thebibliography}{25}

\bibitem{Aberg06}J. {\AA}berg, arXiv:0612146v1 (2006).

\bibitem{Baumgratz13}T. Baumgratz, M. Cramer, and M. B. Plenio, \emph{Phys. Rev. Lett.} \textbf{113}, 140401 (2014).

\bibitem{Streltsov17}A. Streltsov, G. Adesso, and M. B. Plenio, arXiv:1609.02439v3 (2016).

\bibitem{Hu17}M.-L. Hu, X. Hu, J.-C. Wang, Y. Peng, Y.-R. Zhang, and H. Fan, arXiv:1703.01852 (2017).

\bibitem{Modi10}K. Modi, T. Paterek, W. Son, V. Vedral, and M. Williamson, \prl \textbf{104}, 080501 (2010).

\bibitem{Streltsov15}A. Streltsov, U. Singh, H. Shekhar Dhar, M. N. Bera, and G. Adesso, \prl \textbf{115}, 020403 (2015).

\bibitem{Chitambar16}E. Chitambar, A. Streltsov, S. Rana, M. N. Bera, G. Adesso, and M. Lewenstein, \emph{Phys. Rev. Lett.} \textbf{116}, 070402 (2016).

\bibitem{Radhakrishnan16} C. Radhakrishnan, M. Parthasarathy, S. Jambulingam, and T. Byrnes, \prl \textbf{116}, 150504 (2016).

 \bibitem{Ma16}J. J. Ma, B. Yadin, D. Girolami, V. Vedral, and Mile Gu, \emph{Phys. Rev. Lett.} \textbf{116} 160407 (2016).

 \bibitem{Adesso15}T. R. Bromley, M. Cianciaruso, and G. Adesso, \emph{Phys. Rev. Lett.} 114 210401 (2015).

\bibitem{Yao15}Y. Yao, X. Xiao, L.Ge, and C. P. Sun, \emph{Phys. Rev. A} 92 022112 (2015).

 \bibitem{Hu17a}M.-L. Hu, and H. Fan, \pra \textbf{95} 052106 (2017).

\bibitem{Xi15}Z. Xi, Y. M. Li, and H. Fan, \emph{Sci. Rep.} 5. 10922 (2015).

\bibitem{Horodecki05}M. Horodecki, P. Horodecki, R. Horodecki, J. Oppenheim, A. Sen(De), U. Sen, and B. Synak-Radtke, \pra \textbf{71}, 062307 (2005).

\bibitem{Zurek01}H. Ollivier, and W. H. Zurek, \emph{Phys. Rev. Lett.} \textbf{88}, 017901 (2001).

\bibitem{Brodutch12}A.Brodutch, and K. Modi, Quantum Inf. Comput. \textbf{12}, 0721 (2012).

\bibitem{Datta08}A. Datta, Ph.D. thesis, University of NewMexico, arXiv:0807.4490v1 (2008).

\bibitem{Luo08}S.L. Luo, \pra  \textbf{77}, 042303 (2008).

\bibitem{Modi12}K. Modi, A. Brodutch, H. Cable, T. Paterek, and V. Vedral, \rmp \textbf{84}, 1655 (2012).

\bibitem{Wu09}S. Wu, U.V. Poulsen, and K. M{\o}mer, \pra \textbf{80}, 032319 (2009).

\bibitem{Kumar17}A. Kumar, \pra \textbf{96}, 012332 (2017).

\bibitem{Girolami14}D. Girolami, \emph{Phys. Rev. Lett.} \textbf{113}, 170401 (2014).

\bibitem{Cubitt03}T. S. Cubitt, F. Verstraete, W. D\"{u}r, and J. I. Cirac, \prl \textbf{91}, 037902 (2003).

\bibitem{Streltsov12}A. Streltsov, H. Kampermann, and D. Bru{\SS}, \prl \textbf{108}, 250501 (2012).

\bibitem{Chuan12}T. K. Chuan, J. Maillard, K. Modi, T. Paterek, M. Paternostro, and M. Piani, \prl \textbf{109} 070501 (2012).

\bibitem{Bu17}K.F. Bu, L. Li,  A. K. Pati, S.-M. Fei, and J. D. Wu, arXiv:1710.08517v1 (2017).

\end{thebibliography}
\end{document}